\documentclass[10pt, journal]{IEEEtran}
\IEEEoverridecommandlockouts
\usepackage[english]{babel}
\usepackage{blindtext}
\usepackage{algorithm}
\usepackage{algorithmic}
\usepackage{graphicx}
\usepackage{subfigure}
\usepackage{amsmath}
\usepackage{amssymb}
\usepackage{cases}
\usepackage{setspace}
\usepackage{multirow}
\newtheorem{theorem}{\textbf{Theorem}}
\newtheorem{definition}{\textbf{Definition}}

\begin{document}
\title{Primary Channel Gain Estimation for Spectrum Sharing in Cognitive Radio Networks}

\author{\authorblockN{Lin Zhang, Guodong Zhao, Wenli Zhou, Liying Li, Gang Wu, \\Ying-Chang Liang, and Shaoqian Li}
\thanks{Lin Zhang, Guodong Zhao, Wenli Zhou, Gang Wu, Ying-Chang Liang, and Shaoqian Li are with the National Key Lab of Science and Technology on Communications, University of Electronic Science and Technology of China, Chengdu, China, emails: linzhang1913@gmail.com, gdngzhao@gmail.com, di$\_$di$\_$zhou@163.com, wugang99@uestc.edu.cn, liangyc@ieee.org, and lsq@uestc.edu.cn; Liying Li is with School of Automation, University of Electronic Science and Technology of China, Chengdu, China, email: liyingli0815@gmail.com; Guodong Zhao and Liying Li are also with Department of Electrical and Computer Engineering, Lehigh University, Bethlehem, PA, USA.}
}
\maketitle

\thispagestyle{empty}

\begin{abstract}
In cognitive radio networks, the channel gain between primary transceivers, namely, primary channel gain,
is crucial for a cognitive transmitter (CT) to control the transmit power and achieve spectrum sharing. Conventionally, the primary channel gain is estimated in the primary system and thus unavailable at the CT. To deal with this issue, two estimators are proposed by enabling the CT to sense primary signals. In particular, by adopting the maximum likelihood (ML) criterion to analyze the received primary signals, a ML estimator is first developed. After demonstrating the high computational complexity of the ML estimator, a median based (MB) estimator with proved low complexity is then proposed. Furthermore, the estimation accuracy of the MB estimation is theoretically characterized. By comparing the ML estimator and the MB estimator from the aspects of the computational complexity as well as the estimation accuracy, both advantages and disadvantages of two estimators are revealed. Numerical results show that the estimation errors of the ML estimator and the MB estimator can be as small as $0.6$ dB and $0.7$ dB, respectively.
\end{abstract}


\section{Introduction}

Cognitive radio technique is a promising candidate to deal with the spectrum
shortage problem in the wireless communication \cite{Haykin},
\cite{Y_C_Liang_overview}. By coexisting with primary users on the under-utilized licensed spectrum, cognitive users enhance the utilization efficiency of the spectrum meanwhile leverage the cognitive throughput. In general, cognitive users are able to coexist with primary users in two ways. One is
\emph{opportunistic spectrum access} (OSA) \cite{Q_Zhao} and the other is
\emph{spectrum sharing} (SS)
\cite{Cellular_SS}, \cite{5G_SS}. In OSA, cognitive users are allowed to access the spectrum only if the
spectrum is idle, and have to free the spectrum as soon as possible once the
spectrum is re-occupied. In SS, cognitive users are allowed to access the spectrum
even when the spectrum is occupied, provided that the co-channel
interference inflicted to the \emph{primary receiver} (PR) does not
violate a maximum interference power constraint, namely, \emph{interference temperature}. Therefore, compared with OSA, SS is able to exploit more spectrum opportunities and obtain higher spectrum utilization efficiency, both of which boost the cognitive throughput \cite{Capacity_ITC_1}, \cite{Capacity_ITC_2}, \cite{Capacity_ITC_3}.

In recent years, SS has been studied extensively
\cite{Interference_T_1}, \cite{Interference_T_2}, \cite{Interference_T_3}. In these literature, the
interference temperature of the primary system is usually assumed to be
known to the \emph{cognitive transmitter} (CT), such that the CT is able to explicitly control the transmit power and protect primary transmissions. However, to obtain the interference temperature, the CT needs a backhaul link from the primary system. Then, the primary system can transmit the information of the interference temperature to the CT. In practice, there may not exist any backhaul link between the two systems. Thus, it is challenging for the CT to obtain the interference temperature and achieve SS in such a situation.

In fact, the calculation of the interference temperature is highly related to the channel gain between primary transceivers, namely, primary channel gain. Specifically, within a \emph{quality of service} (QoS) guaranteed primary system, the \emph{primary transmitter} (PT) automatically adapts its transmit power to satisfy a target \emph{signal to interference plus noise ratio} (SINR or SNR) at the PR or equivalently a target transmission rate. A large primary channel gain means that the target QoS of a primary transmission can be easily satisfied, even when the transmit power of primary signals is small. Under the maximum transmit power constraint at the PT, the primary transmission is able to tolerate a strong interference signal meanwhile achieve the target QoS by increasing the transmit power of primary signals. This leads to a large interference temperature and contributes to a high cognitive throughput. On the contrary, if the primary channel gain is small, a large transmit power of primary signals is required to satisfy the target QoS. Then, to achieve the target QoS of the primary transmission, only a weak interference signal can be tolerated, even when the PT works with the maximum transmit power. This leads to a small interference temperature and reduces the cognitive throughput. Therefore, the primary channel gain is very important in calculating the interference temperature. (The detailed mathematical calculation of the interference temperature with the primary channel gain can be found in Appendix A.)

Conventionally, the primary channel gain is estimated in the primary system. In particular, the PT transmits a training signal to the PR through the primary channel. The PR extracts the information of the primary channel gain from the received training signal, and calculates the interference temperature. But, the CT cannot obtain the primary channel gain. To deal with this issue, we propose new methods for the CT to estimate the primary channel gain, such that the CT is able to calculate the interference temperature and achieve SS. We note that there exists an implicit relation between primary signals and the primary channel gain. As a consequence, it is possible for the CT to exploit the relation to estimate the primary channel gain. In principle, within a QoS guaranteed primary system, primary signals are carefully designed based on the primary channel gain. In particular, if the primary channel gain is strong, the PT is able to satisfy the target QoS with a low transmit power. Otherwise, the PT needs to increase its transmit power to compensate for the target QoS. In other words, primary signals contain some information of the primary channel gain. Thus, it becomes possible for the CT to obtain the primary channel gain.

In this paper, we develop two estimators, i.e., a high-complexity \emph{maximum likelihood} (ML) estimator and a low-complexity \emph{median based} (MB) estimator, for the CT to obtain the primary channel gain. Numerical results show that the estimation errors of the ML estimator and the MB estimator can reach as small as $0.6$ dB and $0.7$ dB, respectively. Meanwhile, the ML estimator outperforms the MB estimator in terms of the estimation error if the SNR of the received primary signals at the CT is no smaller than $4$ dB. Otherwise, the MB estimator is superior to the ML estimator from the aspect of both the computational complexity and the estimation accuracy.

To our best knowledge, this is the first work that considers unknown interference temperature at the CT and estimates the primary channel gain for SS. The main contributions of this paper are as follows:

{$\bullet$} By enabling the CT to sense primary signals and adopting the ML criterion to analyze the received primary signals, we develop a ML estimator for the CT to obtain the primary channel gain. In particular, the ML estimator is obtained by solving a nonlinear equation. To shed more light on the estimator design, we study the property of the nonlinear equation and develop a bisection method to solve it. In addition, we analyze the computational complexity of the ML estimator.

{$\bullet$} After demonstrating the high computational complexity of the ML estimator, we develop a MB estimator with proved low complexity. By denoting $K$ as the number of the received primary signals, we derive both lower and upper bounds of an estimation with a certain probability. In particular, the probability is expressed as a function of $K$ and monotonously increases as $K$ grows. We also study the performance limit of the MB estimator when $K$ grows to the infinity. Furthermore, we analyze the computational complexity of the MB estimator.

{$\bullet$} By comparing the ML estimator and the MB estimator from the aspects of the computational complexity as well as the estimation accuracy, both advantages and disadvantages of two estimators are revealed. Numerical results verify our theoretical analysis.

\section{System Model}
Fig. {\ref{Fig.1}} provides the system model, which consists of a PT, a PR, and a CT. In particular, the PT is transmitting data to the PR on a wireless channel. Meanwhile, the CT intends to estimate the primary channel gain between the PT and the PR for SS. In what follows, we will present the channel model and signal model in the considered system, respectively.


         \begin{figure}
            \centering
            \includegraphics[scale=0.6]{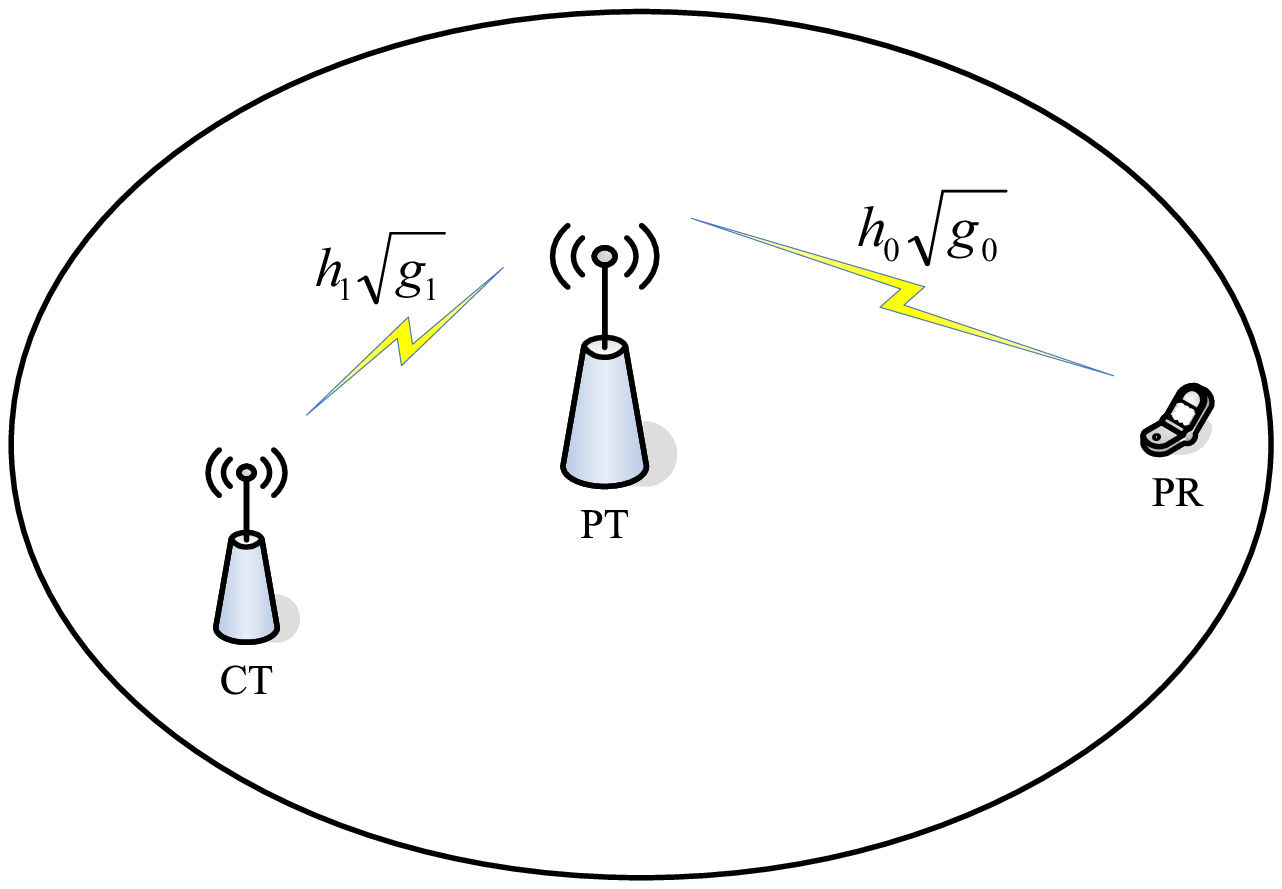}
            \caption{System model, which consists of a PT, a PR, and a CT. In particular, the PT is transmitting data to the PR on a wireless channel. Meanwhile, the CT intends to estimate the primary channel gain between the PT and the PR for SS.}
            \label{Fig.1}
        \end{figure}
\subsection{Channel Model}
Block fading channels are considered among three users. Specifically, if we denote $h_0$ ($h_1$) and $g_0$ ($g_1$) as the small-scale block fading and the large-scale channel gain coefficients between the PT and PR (CT), the channel between the PT and PR (CT) is $h_0\sqrt{g_0}$ ($h_1\sqrt{g_1}$). On one hand, $|h_i|$ ($i=0, \ 1$) follows a Rayleigh distribution with unit mean. $h_i$ ($i=0, \ 1$) remains constant within each block and varies independently among different blocks. On the other hand, $g_i$ ($i=0, \ 1$) is determined by the path-loss model. If we adopt the path loss model \cite{3GPP_channel_model}
\begin{equation}
    P_l(d_i)=128 + 37.6\log_{10}(d_i), \ \ \text{for} \ \ d_i \geq
    0.035 \ \text{km},
    \label{Path_loss_model}
\end{equation}
where $d_i$ (km) is the distance between two transceivers, the large-scale channel gain $g_i$ is
   \begin{equation}
    g_i=10^{-12.8}d^{-3.76}_i, \ \ \text{for} \ \ d_i \geq
    0.035 \ \text{km},
    \label{g_i}
\end{equation}
and remains constant all the time for a given distance $d_i$.

Thus, the CT needs to estimate the primary channel gain $g_0$ from the PT to PR for SS. Since we focus on the sensing phase for the CT to estimate the primary channel gain, we will not discuss the transmission phase in the rest of this paper.

\subsection{Signal Model}

  \subsubsection{Signal Model from the PT to PR}
Denote $x_{p}$ as the primary signal with unit power, i.e., $|x_{p}|^2=1$. If the PT transmits the primary signal with power $p_0$, the received signal at the PR in block $k$ is
 \begin{equation}
 y_p(k)= h_0(k)\sqrt {g_0p_0(k)} x_p(k) + n_p(k),
 \label{y_p}
 \end{equation}
where $n_p$ represents the AWGN at the PR with zero mean and variance $\sigma^2$. Then, the SNR of the received primary signal at the PR is
 \begin{equation}
\gamma_{p}(k)= \frac{|h_0(k)|^2g_0p_0(k)}{\sigma^2}.
 \label{gamma_p}
 \end{equation}

We further consider that the PT and PR adopt \emph{close loop
power control} (CLPC) to provide QoS guaranteed wireless communication \cite{LZ}, \cite{Rui}, \cite{Null_space_1}, \cite{Null_space_2}. That means, the PT automatically adjusts its
transmit power to meet a certain target SNR $\gamma_T$ at the PR.
Then, PT's transmit power satisfies
\begin{equation}
p_0(k)= \frac{{\gamma _T\sigma^2}}{{|h_0(k)|^2{g_0}}}.
\label{p_0}
\end{equation}

\subsubsection{Signal Model from the PT to CT}
In the meantime, the received primary signal at the CT in block $k$ is
  \begin{equation}
y_c(k) = h_1(k)\sqrt {{g_1}{p_0(k)}} {x_{p}}(k) + n_c(k),
 \label{y_c}
 \end{equation}
where $n_c$ is the AWGN at the CT with zero mean and variance $\sigma^2$. Then, the SNR of the received primary signal at the CT is
 \begin{equation}
\gamma_c(k) = \frac{{|h_1(k)|^2{g_1}{p_{0}(k)}}}{{\sigma^2}}.
 \label{gamma_c}
 \end{equation}

By substituting (\ref{p_0}) into (\ref{gamma_c}), $\gamma_c(k)$ in (\ref{gamma_c}) can be rewritten as
\begin{equation}
{\gamma _{c}}(k) = \frac{{{\gamma _T}{g_1}}}{{{g_0}}} \frac{{|h_1(k)|^2}}{{|h_0(k)|^2}}.
\label{gamma_c_1}
\end{equation}

\section{Maximum Likelihood (ML) Estimator}
In this section, we develop a ML estimator to obtain the primary channel gain $g_0$ between the PT and the PR. In what follows, we provide the basic principle of the estimator followed by the estimator design and analysis.

\subsection{Basic Principle}
As aforementioned, within a QoS guaranteed primary system, primary signals are carefully designed based on the primary channel gain. In particular, if the primary channel gain is strong, the PT is able to satisfy the target QoS with a small transmit power of primary signals. Otherwise, the PT increases its transmit power to compensate for the target QoS. In other words, primary signals contain some information of the primary channel gain. Thus, the CT can obtain the primary channel gain by sensing primary signals.

By sensing the received primary signals, the CT can obtain the SNR of the received primary signal as shown in (\ref{gamma_c_1}). From (\ref{gamma_c_1}), each SNR of the received primary signal at the CT is highly related to the primary channel gain $g_0$. Then, it is possible for the CT to measure the SNR of the received primary signal and estimate $g_0$. However, it is difficult to obtain $g_0$ directly from (\ref{gamma_c_1}). This is because, each SNR in (\ref{gamma_c_1}) is also affected by random small-scale fadings and varies independently among different blocks. Alternatively, the CT can measure different SNRs of primary signals in multiple blocks and utilize the distribution knowledge of the small-scale fadings to estimate $g_0$. We note that the ML criterion is able to efficiently extract the common information from multiple data and performs well for parameter estimations \cite{ML_estimation}. Thus, we adopt the ML criterion and develop a ML estimator for the CT to obtain the primary channel gain $g_0$ between the PT and PR.

\subsection{Estimator Design}

By removing the block index $k$ in (\ref{gamma_c_1}) and writing (\ref{gamma_c_1}) into the form of dB, we have
\begin{equation}
{\gamma _{c,dB}}= {\gamma _{T,dB}} + g_{1,dB} - g_{0,dB} + 10{\log _{10}}\phi,
\label{gamma_c_dB}
\end{equation}
where the subscript dB of a parameter is the unit of the parameter, and the random variable $\phi$ is defined as $\phi=|h_1|^2/|h_0|^2$.

Since $\phi$ is a random variable, ${\gamma _{c,dB}}$ in (\ref{gamma_c_dB}) is also a random variable. Then, the \emph{cumulative density function} (CDF) of ${\gamma _{c,dB}}$ can be expressed as
\begin{align}\nonumber
F_{\Gamma_{c,dB}}\left(\gamma_{c,dB}\right) =& \Pr\left\{ {\gamma _{T,dB}} + g_{1,dB} - g_{0,dB} \right.\\\nonumber
&\left. \ \ \ + 10{\log _{10}}\phi\leq \gamma_{c,dB}\right\}\\\nonumber
=& \!\Pr\!\left\{\!\phi \!\leq \!10^{\frac{\gamma_{c,dB}- {\gamma _{T,dB}}-g_{1,dB} + g_{0,dB}}{10}}  \right\}\\
=&F_{\Phi}\left(10^{\frac{\gamma_{c,dB}- {\gamma _{T,dB}}-g_{1,dB} + g_{0,dB}}{10}} \right),
\label{gamma_c_dB_cdf}
\end{align}
where $F_{\Phi}(\cdot)$ denotes the CDF of $\phi$.

Since $|h_i|$ ($i=0, \ 1$) follows a Rayleigh distribution with unit mean, the CDF of $\phi=|h_1|^2/|h_0|^2$ is \cite{Joint_PDF}
\begin{equation}
F_{\Phi}(\phi) = \frac{{\phi}}{{{{1 + \phi}}}}.
\label{phi_cdf}
\end{equation}

Substituting (\ref{phi_cdf}) into (\ref{gamma_c_dB_cdf}), we have the CDF of ${\gamma _{c,dB}}$ as
\begin{equation}
F_{\Gamma_{c,dB}}\left(\gamma_{c,dB}\right)=\frac{{10^{\frac{\gamma_{c,dB}- {\gamma _{T,dB}}-g_{1,dB} + g_{0,dB}}{10}}}}{{{{1 + 10^{\frac{\gamma_{c,dB}- {\gamma _{T,dB}}-g_{1,dB} + g_{0,dB}}{10}}}}}}.
\label{gamma_c_dB_cdf_1}
\end{equation}


By taking the derivation of $F_{\Gamma_{c,dB}}\left(\gamma_{c,dB}\right)$ in terms of $\gamma_{c,dB}$, we have the \emph{probability density function} (PDF) of ${\gamma _{c,dB}}$ as
\begin{align}\nonumber
f\left( {{\gamma _{c, dB}}} \right)=& \frac{{\partial F_{\Gamma_{c,dB}}({\gamma _{c,dB}})}}{{\partial {\gamma _{c,dB}}}}\\
=& \frac{{\frac{{\ln 10}}{{10}}{{10}^{\frac{{{\gamma _{T,dB}} + g_{1,dB} - g_{0,dB} - {\gamma _{c,dB}}}}{{10}}}}}}{{{{\left( {1 + {{10}^{\frac{{{\gamma _{T,dB}} + g_{1,dB} - g_{0,dB} - {\gamma _{c,dB}}}}{{10}}}}} \right)}^2}}}.
\label{gamma_c_dB_pdf}
\end{align}

For $K$ independent blocks, the CT is able to measure $K$ independent values of $\gamma_{c,dB}$, i.e., $\gamma _{c, dB}(k)$ ($1\leq k\leq K$). Then, the joint PDF of $\gamma _{c, dB}(k)$ ($1\leq k\leq K$) is
\begin{align}\nonumber
&{f\left( {{\gamma _{c,dB}(1)},{\gamma _{c,dB}(2)},...,{\gamma _{c,dB}(K)}} \right)}\\
=& \prod\limits_{k = 1}^K {\left[ {\frac{{\frac{{\ln 10}}{{10}}{{10}^{\frac{{{\gamma _{T,dB}} + g_{1,dB} - g_{0,dB} - {\gamma _{c,dB}(k)}}}{{10}}}}}}{{{{\left( {1 + {{10}^{\frac{{{\gamma _{T,dB}} + g_{1,dB} - g_{0,dB} - {\gamma _{c,dB}(k)}}}{{10}}}}} \right)}^2}}}} \right]}.
\label{joint_pdf}
\end{align}

Based on the ML criterion, $g_{0,dB}$ can be approximated with the largest probability by the optimal $g^*_{0,dB}$, which maximizes the joint PDF $f\left( {{\gamma _{c,dB}(1)},{\gamma _{c,dB}(2)},...,{\gamma _{c,dB}(K)}} \right)$. Thus, we shall find the optimal $g^*_{0,dB}$ in the following.

Taking the logarithm operation on both sides of (\ref{joint_pdf}), we have
\begin{equation}
\begin{split}
&{\log _{10}}f\left( {{\gamma _{c,dB}}\left( 1 \right),{\gamma _{c,dB}}\left( 2 \right), \ldots ,{\gamma _{c,dB}}\left( K \right)} \right)\\
 =& \sum\limits_{k = 1}^K \left[ {{\log }_{10}}\left( {\ln 10 \times {{10}^{\frac{{{\gamma _{T,dB}} + g_{1,dB} - g_{0,dB} - {\gamma _{c,dB}}\left( k \right)}}{{10}}}}} \right) \right. \\
&\left.- 2{{\log }_{10}}\left( {1 + {{10}^{\frac{{{\gamma _{T,dB}} + g_{1,dB} - g_{0,dB} - {\gamma _{c,dB}}\left( k \right)}}{{10}}}}} \right) - 1 \right].
\end{split}
\label{log_joint_pdf}
\end{equation}

Taking the derivation of (\ref{log_joint_pdf}) in terms of $g_{0,dB}$, we obtain
\begin{align}\nonumber
&\frac{{\partial \left\{ {{{\log }_{10}}f\left( {\gamma _{c,dB} (1),{\gamma _{c,dB}(2)},...,{\gamma _{c,dB}(K)}} \right)} \right\}}}{{\partial {g_{0,dB}}}} \\
= &\sum\limits_{k = 1}^K {\left( {\frac{1}{{10}}\times \frac{{{{10}^{\frac{{{\gamma _{T,dB}} + g_{1,dB} - g_{0,dB}- {\gamma _{c,dB}}}}{{10}}}} - 1}}{{{{10}^{\frac{{{\gamma _{T,dB}} + g_{1,dB} - g_{0,dB} - {\gamma _{c,dB}}}}{{10}}}} + 1}}} \right)}.
\label{a14}
\end{align}

Thus, we can find the optimal $g^*_{0,dB}$ by solving $\frac{{\partial \left\{ {{{\log }_{10}}f\left( {\gamma _{c,dB} (1),{\gamma _{c,dB}(2)},...,{\gamma _{c,dB}(K)}} \right)} \right\}}}{{\partial {g_{0,dB}}}} =0$, i.e.,
\begin{equation}
\sum\limits_{k = 1}^K {\left( {\frac{1}{{10}} \times \frac{{{{10}^{\frac{{{\gamma _{T,dB}} + g_{1,dB} - g_{0,dB} - {\gamma _{c,dB}}}}{{10}}}} - 1}}{{{{10}^{\frac{{{\gamma _{T,dB}} + g_{1,dB} - g_{0,dB} - {\gamma _{c,dB}}}}{{10}}}} + 1}}} \right)}=0.
\label{equation_optimal}
\end{equation}

After obtaining the optimal $g^*_{0,dB}$, we have the ML estimator as
\begin{equation}
\hat{g}_{0,dB}=g^*_{0,dB}.
\label{ML_estimator}
\end{equation}

From (\ref{equation_optimal}), the ML estimator $\hat{g}_{0,dB}$ (or the optimal $g^*_{0,dB}$) is determined by the target SNR $\gamma_{T, dB}$ at the PR, the channel gain $g_{1,dB}$ between the PT and CT, and the SNR $\gamma_{c, dB}$ of the primary signal at the CT. It is notable that, $\gamma_{T, dB}$ can be obtained by the CT through observing the \emph{modulation and coding scheme} (MCS) of the primary signal \cite{D_Tse}. $g_{1,dB}$ is a deterministic function of the distance $d_1$ between the PT and CT, and thus can be calculated at the CT. $\gamma_{c, dB}$ is measured at the CT and also known to the CT. Therefore, the CT is able to solve (\ref{equation_optimal}) and obtain the ML estimator (\ref{ML_estimator}). However, it is difficult to solve (\ref{equation_optimal}) directly, since (\ref{equation_optimal}) is a non-linear equation of $g_{0, dB}$. To deal with this issue, we will develop a bisection method \cite{Binary_search} to solve (\ref{equation_optimal}) in the next part.

\subsection{Bisection Method to Solve (\ref{equation_optimal})}
In this part, we first study the property of (\ref{equation_optimal}) and demonstrate that it is possible to solve (\ref{equation_optimal}) with a bisection method. Then, we develop a bisection method to solve (\ref{equation_optimal}) and obtain the optimal $g^*_{0,dB}$.

To begin with, we denote
\begin{equation}
f_1(g_{0,dB})\!=\!\sum\limits_{k = 1}^K \!{\left( {\frac{1}{{10}} \times \frac{{{{10}^{\frac{{{\gamma _{T,dB}} + g_{1,dB} - g_{0,dB} - {\gamma _{c,dB}}}}{{10}}}} - 1}}{{{{10}^{\frac{{{\gamma _{T,dB}} + g_{1,dB} - g_{0,dB} - {\gamma _{c,dB}}}}{{10}}}} + 1}}} \right)}.
\label{f_1_g_0}
\end{equation}

By taking the derivation of $f_1(g_{0,dB})$ in terms of $g_{0,dB}$, we have
\begin{align} \frac{{\partial {f_1}({g_{0,dB}})}}{{\partial {g_{0,dB}}}} = \sum\limits_{k = 1}^K {\left( {\frac{{ - \ln(10) {{10}^{\frac{{{\gamma _{T,dB}}+g_{1,dB} - g_{0,dB} - {\gamma _{c,dB}}}}{{10}}}}}}{{50{{\left( {1 + {{10}^{\frac{{{\gamma _{T,dB}}+g_{1,dB} - g_{0,dB} - {\gamma _{c, dB}}}}{{10}}}}} \right)}^2}}}} \right)},
\label{derivation_f_1}
\end{align}
which is smaller than or equal to $0$, i.e., $\frac{{\partial {f_1}({g_{0,dB}})}}{{\partial {g_{0,dB}}}} \leq 0$. In other words, ${f_1}({g_{0,dB}})$ monotonically decreases as $g_{0,dB}$ increases. Besides, we observe $f_1(-\infty)=\frac{K}{10}>0$ and $f_1(+\infty)=-\frac{K}{10}<0$. Therefore, (\ref{equation_optimal}) has a unique positive solution and can be efficiently solved by a bisection method.

We provide the detailed bisection method in \textbf{Algorithm \ref{Bisection Method}}. In particular, $g_{0,dB}^{\min}$ and $g_{0,dB}^{\max}$ are the initial values of the bisection method, and denote the minimum value and the maximum value of $g_{0,dB}$, respectively. On one hand, since the PR is in the coverage of the PT, and the maximum coverage radius $d_0=R$ of the PT can be substituted into (\ref{g_i}) to calculate $g_{0,dB}^{\min}$  , i.e., $g_{0,dB}^{\min}=-128-37.6\log_{10}(R)$. For instance, we suppose that the PT is a base station of a cell, the radius of the cell can be known by the CT and be used to calculate $g_{0,dB}^{\min}$. On the other hand, the large-scale channel gain in (\ref{g_i}) requires that the distance between two transceivers is no less than $0.035$ km, $g_{0,dB}^{\max}$ can be calculated by substituting $d_0=0.035$ km into (\ref{g_i}), i.e., $g_{0,dB}^{\max}=-128-37.6\log_{10}(0.035)$. Similarly, if another path-loss model different from (\ref{Path_loss_model}) is adopted, $g_{0,dB}^{\min}$ and $g_{0,dB}^{\max}$ can also be calculated with minor modifications.

\begin{algorithm}[htb]

\algsetup{linenosize=\footnotesize }

\caption{Bisection Method for $g^*_{0,dB}$.}

\label{Bisection Method}

\begin{algorithmic}[1]

\REQUIRE ~~\ \

$g_{0,dB}^{\min}$, $g_{0,dB}^{\max}$, $g_{0,dB}^{\text{mid}}$, and the maximum tolerance error $\nu$;

\ENSURE ~~\ \

\WHILE {$|g_{0,dB}^{\max}-g_{0,dB}^{\min}|>\nu$}

\STATE $g_{0,dB}^{\text{mid} } = \frac{{g_{0,dB}^{\max } + g_{0,dB}^{\min }}}{2}$;

\IF{${f_1}\left( {g_{0,dB}^{{\mathop{\rm mi}\nolimits} d}} \right) {f_1}\left( {g_{0,dB}^{\min }} \right) > 0$}

\STATE $g_{0,dB}^{\min}=g_{0,dB}^{\text{mid}}$;

\ELSE

\STATE $g_{0,dB}^{\max}=g_{0,dB}^{\text{mid}}$;

\ENDIF

\ENDWHILE

\STATE {\bf Return} $g^*_{0,dB}=g_{0,dB}^{\text{mid}}$.

\end{algorithmic}

\end{algorithm}

\subsection{Complexity Analysis}
From the previous parts, the computational complexity of the ML estimator is dominated by solving (\ref{equation_optimal}) with the proposed bisection method in \textbf{Algorithm \ref{Bisection Method}}. Besides,
the computational complexity of the bisection method is $O\left( {{{\log }_2}\frac{{g_{0,dB}^{\max } - g_{0,dB}^{\min }}}{{\nu}}}\right)$ \cite{Binary_search}, \cite{Bisection}, where $\nu$ is the maximum tolerance error of the bisection method in \textbf{Algorithm \ref{Bisection Method}}. Thus, the computational complexity of the ML estimator is $O\left( {{{\log }_2}\frac{{g_{0,dB}^{\max } - g_{0,dB}^{\min }}}{{\nu}}}\right)$.

\section{Median Based (MB) Estimator}
In the previous section, we have developed a ML estimator to obtain an estimation of the primary channel gain. In particular, the ML estimator requires to solve a nonlinear equation and is computationally complicated. In this section, we will present a low complexity estimator. In what follows, we provide the basic principle of the estimator followed by the estimator design and performance analysis.

\subsection{Basic Principle}
To begin with, we provide the definition of the median $x_{\frac{1}{2}}$ of a random variable $X$ as follows,

\begin{definition} For a random variable $X$ with CDF $F_X(x)$, $x\in \mathbb{R}$, if $x_{\frac{1}{2}}$ satisfies both \begin{equation}
F_X(x_{\frac{1}{2}})= \Pr\{X\leq x_{\frac{1}{2}}\}=\frac{1}{2}
\label{median_1}
\end{equation}
and
\begin{equation}
1-F_X(x_{\frac{1}{2}})= \Pr\{X\geq x_{\frac{1}{2}}\}=\frac{1}{2},
\label{median_2}
\end{equation}
$x_{\frac{1}{2}}$ is defined as the median of the random variable $X$.
\end{definition}

Based on Definition 1, we can obtain the median $\gamma_{c, dB, \frac{1}{2}}$ of the random variable $\gamma_{c, dB}$ by letting $F_{\Gamma_{c,dB}}\left(\gamma_{c,dB}\right)$ in (\ref{gamma_c_dB_cdf}) be $\frac{1}{2}$, i.e.,
\begin{equation}
\begin{split}
F_{\Phi}\left(10^{\frac{\gamma_{c,dB}- {\gamma _{T,dB}}-g_{1,dB} + g_{0,dB}}{10}} \right)=\frac{1}{2}.
\label{gamma_c_dB_median_1}
\end{split}
\end{equation}

By substituting (\ref{phi_cdf}) into (\ref{gamma_c_dB_median_1}), we have
\begin{align}\nonumber
&F_{\Phi}\left(10^{\frac{\gamma_{c,dB}- {\gamma _{T,dB}}-g_{1,dB} + g_{0,dB}}{10}} \right)\\ \nonumber
=&\frac{10^{\frac{\gamma_{c,dB}- {\gamma _{T,dB}}-g_{1,dB} + g_{0,dB}}{10}}}{1+10^{\frac{\gamma_{c,dB}- {\gamma _{T,dB}}-g_{1,dB} + g_{0,dB}}{10}}}\\
=&\frac{1}{2}.
\label{gamma_c_dB_median_2}
\end{align}

After solving (\ref{gamma_c_dB_median_2}), the median $\gamma_{c, dB, \frac{1}{2}}$ of the random variable $\gamma_{c, dB}$ can be derived as
\begin{equation}
\begin{split}
\gamma_{c, dB, \frac{1}{2}}= {\gamma _{T,dB}} + g_{1,dB} - g_{0,dB}.
\label{gamma_c_dB_median}
\end{split}
\end{equation}

From (\ref{gamma_c_dB_median}), the median $\gamma_{c, dB, \frac{1}{2}}$ is a function of the primary channel gain $g_{0,dB}$. Thus, if $\gamma_{c, dB, \frac{1}{2}}$ is available to the CT, $g_{0,dB}$ can be directly calculated with (\ref{gamma_c_dB_median}). However, $\gamma_{c, dB, \frac{1}{2}}$ is unknown to the CT. Instead, we will first estimate $\gamma_{c, dB, \frac{1}{2}}$ and then obtain the estimation of $g_{0,dB}$ with (\ref{gamma_c_dB_median}).

\subsection{Estimator Design}
We first give the definition of the sample median $x^s_{\frac{1}{2}}$ of a random variable $X$ as follows,

\begin{definition} For a random variable $X$ with samples $x_m$ ($1\leq m \leq M$), if $x^{s}_{\frac{1}{2}}$ satisfies both \begin{equation} \nonumber
\Pr\{x_m\leq x^{s}_{\frac{1}{2}}\}=\frac{1}{2}
\label{sample_median_1}
\end{equation}
and
\begin{equation} \nonumber
\Pr\{x_m \geq x^{s}_{\frac{1}{2}}\}=\frac{1}{2},
\label{sample_median_2}
\end{equation}
$x^{s}_{\frac{1}{2}}$ is defined as the sample median of the random variable $X$.
\end{definition}

As mentioned in the previous section, for $K$ independent blocks, the CT is able to measure $K$ independent samples of $\gamma_{c, dB}$, i.e., $\gamma _{c, dB}(k)$ ($1\leq k\leq K$). In what follows, we approximate the median $\gamma_{c, dB, \frac{1}{2}}$ with the sample median $\gamma^s_{c, dB, \frac{1}{2}}$ of these $K$ samples. With the approximated $\gamma_{c, dB, \frac{1}{2}}$, $g_{0,dB}$ can be estimated by calculating (\ref{gamma_c_dB_median}).

To begin with, by sorting the $K$ samples in ascending order, the $K$ samples can be relabelled as $\bar{\gamma}_{c, dB}(k)$ ($1\leq k\leq K$), i.e., $\bar{\gamma} _{c, dB}(i) \leq \bar{\gamma} _{c, dB}(j)$ for $1\leq i \leq j \leq K$. Since the sample medians $\bar \gamma^s_{c, dB, \frac{1}{2}}$ of these $K$ samples for odd and even $K$ can be different, we will discuss sample medians for odd and even $K$ separately.

 \subsubsection{For the case that $K$ is odd} When $K$ is odd, the sample median is $\gamma^s_{c, dB, \frac{1}{2}}=\bar \gamma _{c, dB}\left(\frac{K+1}{2}\right)$. Then, the median of $\gamma_{c, dB}$ can be approximated as
\begin{equation}
\gamma_{c, dB, \frac{1}{2}}\approx \bar \gamma _{c, dB}\left(\frac{K+1}{2}\right).
\label{approximate_odd}
\end{equation}

By substituting (\ref{approximate_odd}) into (\ref{gamma_c_dB_median}), we have the MB estimator as
\begin{equation}
\begin{split}
\hat{g}_{0,dB}= {\gamma _{T,dB}} + g_{1,dB}-\bar \gamma _{c, dB}\left(\frac{K+1}{2}\right).
\label{MB_estimator_odd}
\end{split}
\end{equation}

 \subsubsection{For the case that $K$ is even} When $K$ is even, the sample median is between $\bar \gamma_{c, dB}\left(\frac{K}{2}\right)$ and $\bar \gamma _{c, dB}\left(\frac{K}{2}+1\right)$. Then, the median of $\gamma_{c, dB}$ can be approximated as
\begin{equation}
\gamma_{c, dB, \frac{1}{2}}\approx\frac{\bar\gamma _{c, dB}\left(\frac{K}{2}\right)+\bar\gamma _{c, dB}\left(\frac{K}{2}+1\right)}{2}.
\label{approximate_even}
\end{equation}

By substituting (\ref{approximate_even}) into (\ref{gamma_c_dB_median}), we have the MB estimator as
\begin{equation}
\begin{split}
\hat{g}_{0,dB}={\gamma _{T,dB}} +g_{1,dB}-\frac{\bar \gamma _{c, dB}\left(\frac{K}{2}\right)+\bar \gamma _{c, dB}\left(\frac{K}{2}+1\right)}{2}.
\label{MB_estimator_even}
\end{split}
\end{equation}

Consequently, the MB estimator can be summarized as (\ref{MB_estimator}) on the top of the next page.
\begin{figure*}[t!]
\begin{eqnarray}
\hat{g}_{0,dB}=\left\{
\begin{split}
&{\gamma _{T,dB}} + g_{1,dB}-\bar \gamma _{c, dB}\left(\frac{K+1}{2}\right), \quad \quad \quad \quad \ \ \text{for} \ K \ \text{is odd}, \\
&{\gamma _{T,dB}} + g_{1,dB}-\frac{\bar \gamma _{c, dB}\left(\frac{K}{2}\right)+\bar \gamma _{c, dB}\left(\frac{K}{2}+1\right)}{2}, \ \text{for} \ K \ \text{is even}.
\end{split}
\right.
\label{MB_estimator}
\end{eqnarray}
\hrule
\end{figure*}


From (\ref{MB_estimator}), the MB estimator is determined by the target SNR $\gamma _{T,dB}$, the channel gain $g_{1,dB}$ from the PT to the CT, and the measured SNRs at the CT, all of which are available to CT. Thus, the estimation of $g_{0,dB}$ can be directly calculated with (\ref{MB_estimator}). In other words, the computational complexity of the MB estimator in (\ref{MB_estimator}) is $O(1)$.

\begin{theorem} The estimation in (\ref{MB_estimator}) approaches $g_{0,dB}$ as $K$ grows to the infinity, i.e., for any $\mu>0$,
\begin{equation}
\underset{K \to \infty}{\lim} \Pr\left\{ \left| \hat{g}_{0,dB}-g_{0,dB} \right| < \mu \right\} = 1.
\label{Theorem_2}
\end{equation}
\end{theorem}

\begin{proof}
The detailed proof of this Theorem is provided in Appendix C.
\end{proof}

Theorem $1$ indicates that the estimator in (\ref{MB_estimator}) is consistent. In other words, if the number $K$ of the measured SNRs at the CT is large enough, the estimation error is negligible. However, a large $K$ means that the CT needs to measure more primary signals. This increases the computational complexity of the MB estimator. Thus, a proper $K$ is required to balance the estimation accuracy and the computational complexity in practical situations.

\subsection{Comparison between of the ML Estimator and the MB Estimator}
In this part, we compare the ML estimator and the MB estimator from two aspects, i.e., computational complexity and estimation accuracy.

\subsubsection{Computational complexity comparison} As mentioned above, the computational complexity of the ML estimator and the MB estimator are $O\left( {{{\log }_2}\frac{{g_{0,dB}^{\max } - g_{0,dB}^{\min }}}{{\nu}}}\right)$ and $O(1)$, respectively. Thus, the MB estimator is much simpler than the ML estimator.

\subsubsection{Estimation error comparison} In principle, the ML estimator utilizes all the available samples of $\gamma_{c, dB}$, i.e., $\gamma_{c, dB}(k)$ ($1\leq k \leq K$), and outputs an estimation of $g_{0,dB}$. This is different from the MB estimator, which only utilizes the sample median to estimate $g_{0,dB}$. Ideally, the more samples one estimator utilizes, the more accurate the estimation is. In fact, each sample of $\gamma_{c, dB}$ is physically measured at the CT and disturbed by the noise. Thus, each sample contains both the information of $g_{0,dB}$ and noise. In particular, if each measured SNR sample of $\gamma_{c, dB}$ is large, i.e., the conveyed information of each sample is much more than the contained noise, estimators are able to extract more knowledge of $g_{0,dB}$ from more samples, and obtain more accurate estimations. Otherwise, estimators will be more confused by more samples, and thus output less accurate estimations. Therefore, the ML estimator is expected to outperform the MB estimation in terms of the estimation accuracy when the measured SNRs at the CT are large. Otherwise, the MB estimator is superior to the ML estimator. This is verified through numerical results.

\section{Numerical results}
In this section, we provide the numerical results to demonstrate the performance of the proposed ML estimator and MB estimator. Here, we adopt the system model in Section II, where the radius of the PT's coverage is $R=0.5$ km. Besides, we assume that the power of the AWGN is $\sigma^2=-114$ dBm, and the target SNR of the PR is $\gamma_{T,dB}=10$ dB, and the number of samples to measure a SNR at the CT within each block is $J=100$, and the tolerance error in \textbf{Algorithm 1} is $\nu=0.1$. Furthermore, $10^4$ Monte Carlo trails are conducted for each curve.


To begin with, we define the estimation error as
\begin{equation}\label{a31}
\varepsilon = \left| \hat g_{0,dB} - g_{0,dB} \right|.
\end{equation}

\begin{figure}[t!]
\centering
\includegraphics[scale=0.5]{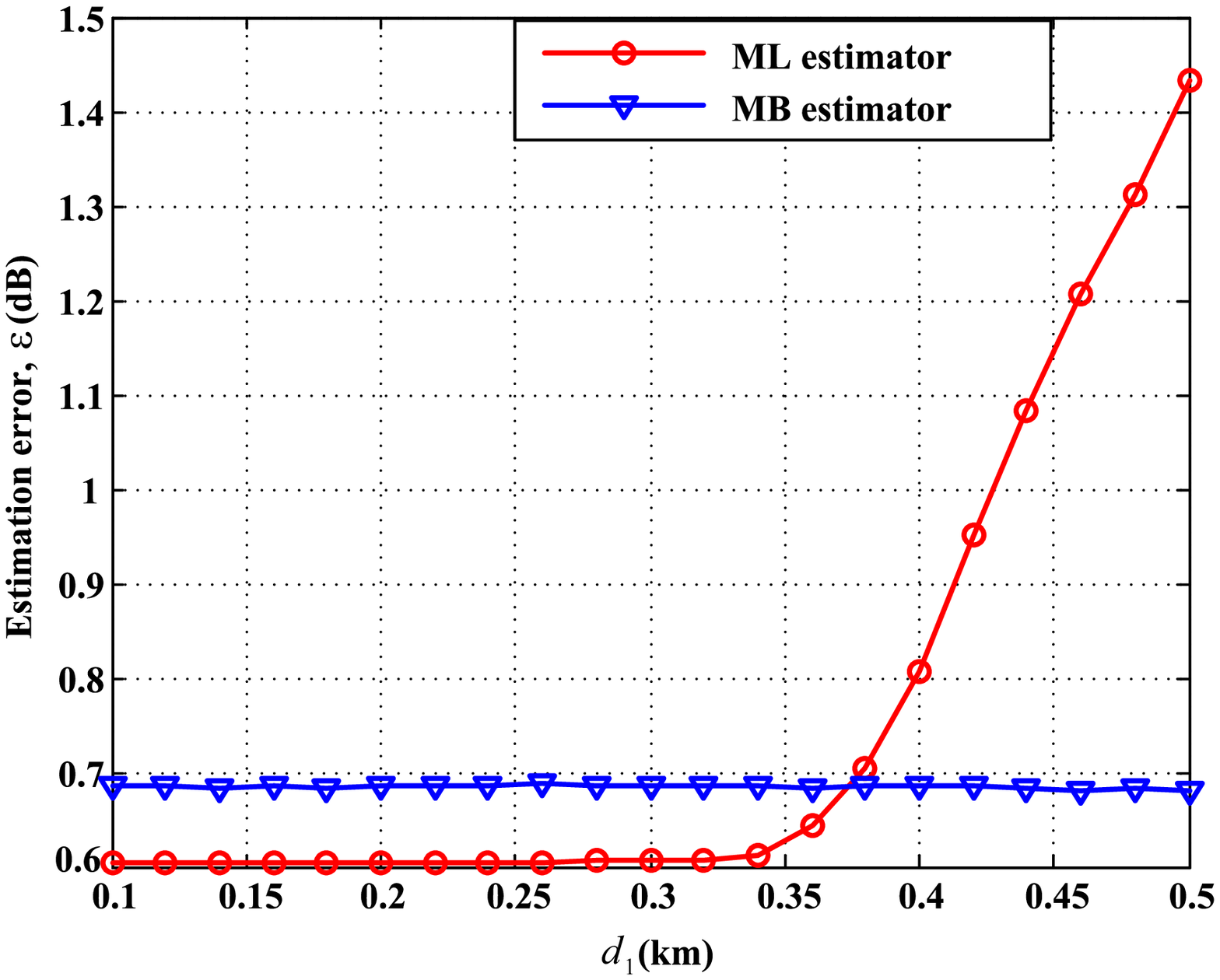}
\caption{The estimation errors with different distances $d_1$ between the PT and CT. In particular, the distance $d_0$ between the PT and PR is $0.25$ km.}
\label{error_d1}
\end{figure}

Fig. \ref{error_d1} illustrates the estimation error $\varepsilon$ when the distance $d_1$ between the PT and CT grows from $0.1$ km to $0.5$ km. In particular, the distance $d_0$ between the PT and PR is $0.25$ km. From this figure, $\varepsilon$ of the ML estimator remains at around $0.6$ dB as $d_1$ grows from $0.1$ km to $0.35$ km and increases from around $0.6$ dB to $1.45$ dB as $d_1$ grows from $0.35$ km to $0.5$ km. Meanwhile, we observe that $\varepsilon$ of the MB estimator remains at around $0.68$ dB for $d_1\leq 0.5$ km. By comparing the estimation errors of the ML estimator and the MB estimator, the ML estimator outperforms the MB estimator for $d_1\leq 0.37$. And the MB estimator is superior to the ML estimator for $d_1 \geq 0.37$.

\begin{figure}[t!]
\centering
\includegraphics[scale=0.5]{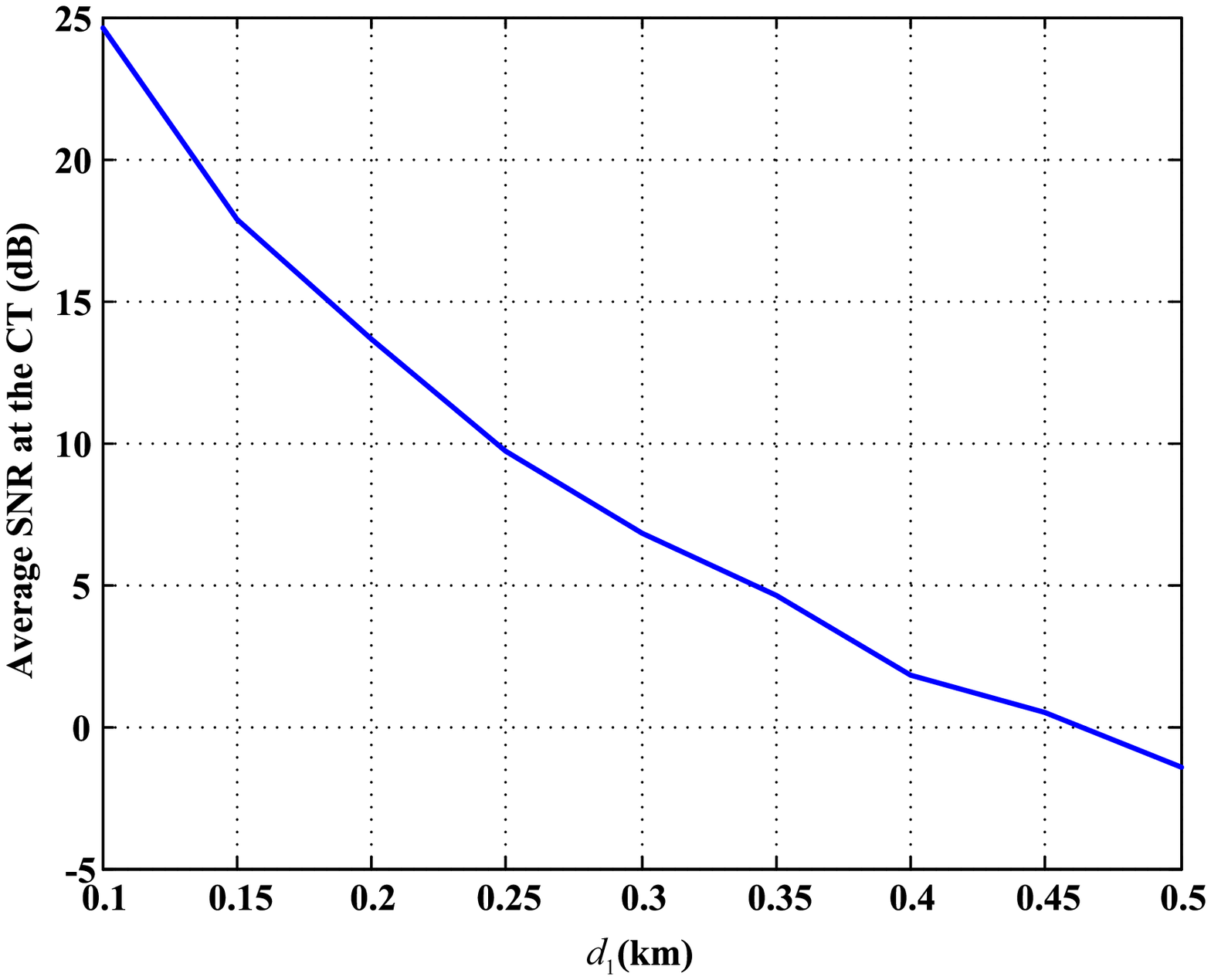}
\caption{Average measured SNR at the CT with different distances $d_1$ between the PT and CT. In particular, the distance $d_0$ between the PT and PR is $0.25$ km.}
\label{SNR_d1}
\end{figure}

Note that, the estimation error of an estimator is affected by two factors, i.e., the number $K$ of the measured SNRs at the CT and the measure error of each SNR. Since $K$ in Fig. \ref{error_d1} is fixed by $K=100$ as $d_1$ grows, the incurred estimation error by $K$ is also fixed. Then, the variation of the estimation error in this figure is caused by the measure error of each SNR. We analyze the variation of the estimation error as follows: Since the distance $d_0$ between the PT and PR is fixed at $d_1=0.25$ km, the average transmit power of the PT remains constant to guarantee the target SNR at the PR. As $d_1$ grows from $0.1$ km to $0.5$ km, the channel gain $g_{1,dB}$ is degraded. Then, the average SNR of the measured primary signals at the CT decreases from around $25$ dB to around $-1$ dB as shown in Fig. \ref{SNR_d1}. This increases the measure error of each SNR at the CT. By adopting these measured SNRs to estimate the primary channel gain $g_{0,dB}$, the estimation error is increased.

Since the estimation error of the ML estimator for $d_1\leq 0.35$ km in Fig. \ref{error_d1} almost remains constant, the estimation error of the ML estimator caused by the measure error of each SNR is negligible for $d_1\leq 0.35$ km, i.e., the average SNR at the CT is accordingly no less than $5$ dB from Fig. \ref{SNR_d1}. Thus, the estimation error of the ML estimator is dominated by the the number $K$ of the measured SNRs at the CT when the average SNR at the CT is no less than $5$ dB. For $d_1\geq 0.35$, i.e., the average SNR at the CT is accordingly smaller than $5$ dB, the estimation error of the ML estimator increases as $d_1$ grows. Thus, the estimation error of the ML estimator is dominated by the the number $K$ of the measured SNRs at the CT as well as the measure errors of each SNR when the average SNR at the CT is smaller than $5$ dB. Furthermore, since the estimation error of the MB estimator in Fig. \ref{error_d1} remains constant for $d_1\leq 0.5$ km, i.e., the average SNR at the CT is accordingly no less than $-1$ dB from Fig. \ref{SNR_d1}, the estimation error caused by the measure error of each SNR is negligible. Thus, the estimation error of the MB estimator is dominated by the number $K$ of the measured SNRs at the CT when the average SNR at the CT is no less than $-1$ dB. Meanwhile, the MB estimator is more robust than the ML estimator respect to the measure error of each SNR. By comparing the estimation errors of the ML estimator and the MB estimator, the ML estimator outperforms the MB estimator when the average SNR at the CT is no less than $4$ dB. And the MB estimator is superior to the ML estimator when the average SNR at the CT is smaller than $4$ dB. This also verified our analysis in Section IV.

\begin{figure}[t!]
\centering
            \includegraphics[scale=0.5]{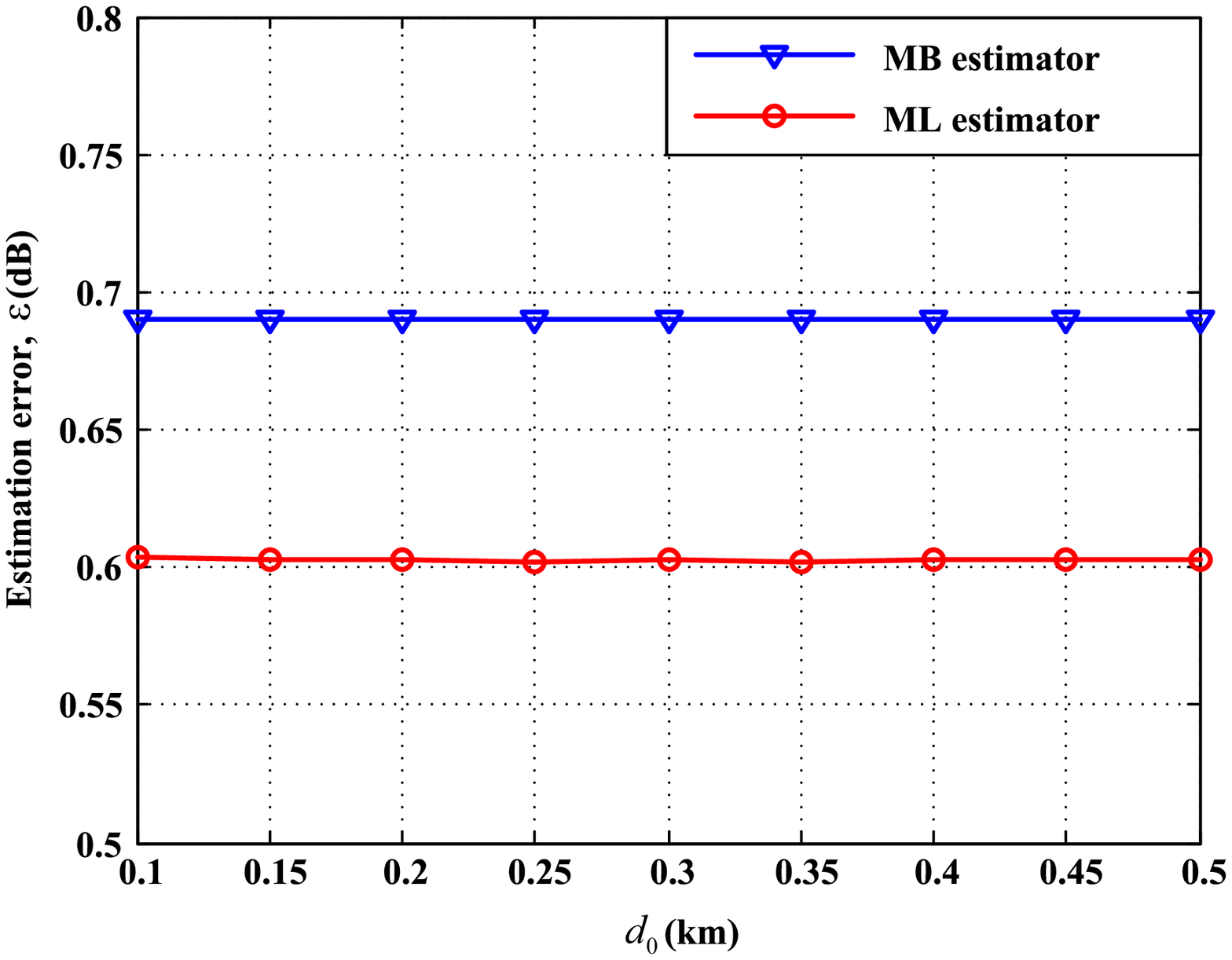}
            \caption{The estimation errors $\varepsilon$ with different distances $d_0$ between the PT and PR. In particular, the distance $d_1$ between the PT and CT is $0.1$ km.}
            \label{error_d0}
        \end{figure}

 \begin{figure}[t!]
\centering
\includegraphics[scale=0.5]{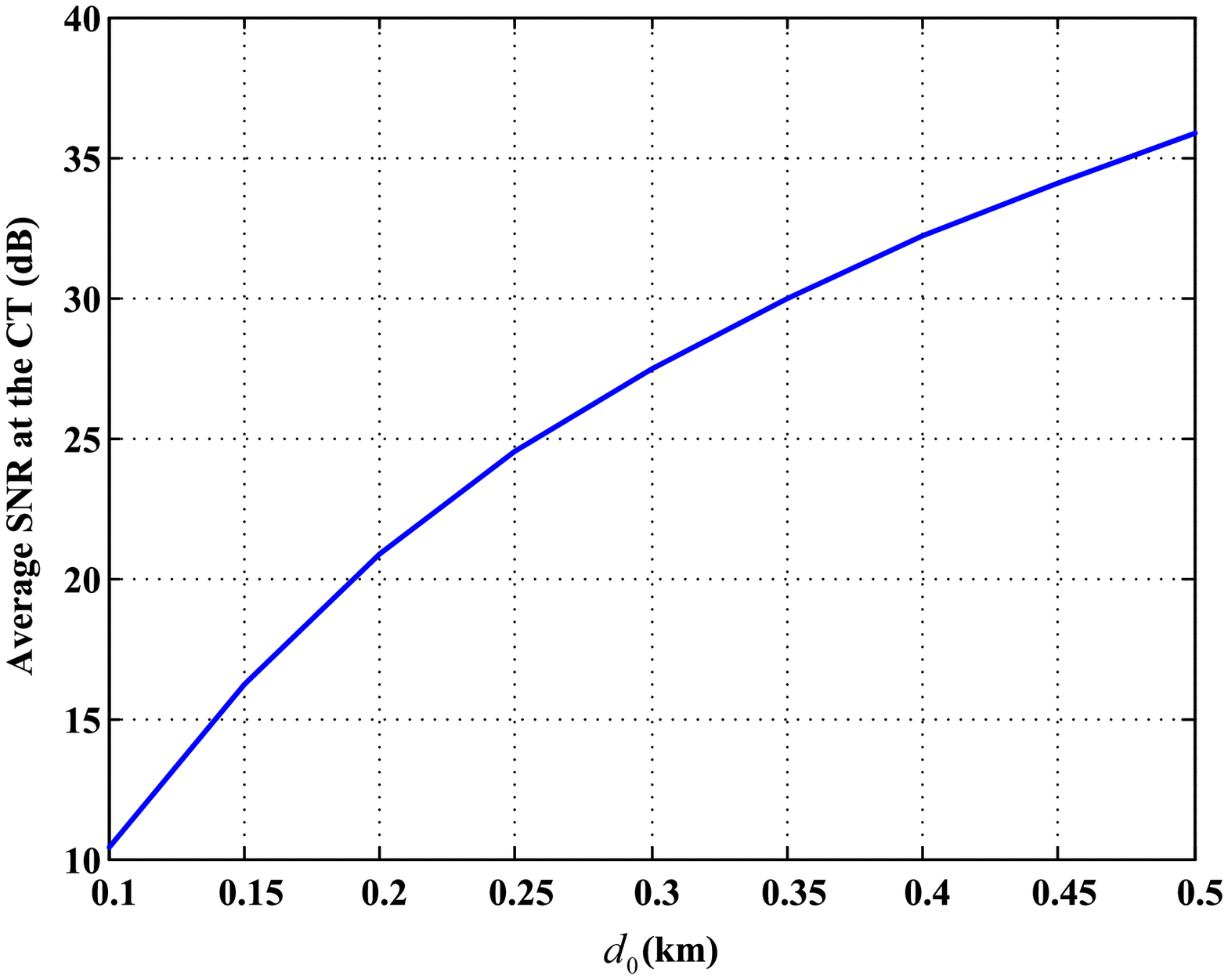}
\caption{Average measured SNR at the CT with different distances $d_0$ between the PT and PR. In particular, the distance $d_0$ between the PT and PR is $0.1$ km.}
\label{SNR_d0}
\end{figure}

Fig. \ref{error_d0} shows the estimation error $\varepsilon$ when the distance $d_0$ between the PT and PR grows from $0.1$ km to $0.5$ km. In particular, the distance $d_1$ between the PT and CT is $0.1$ km. From this figure, $\varepsilon$ of the ML estimator and the MB estimator remain at around $0.6$ dB and $0.7$ dB respectively as $d_0$ grows from $0.1$ km to $0.5$ km. By comparing $\varepsilon$ of the ML estimator and the MB estimator, the ML estimator outperforms the MB estimator. The reasons are follows: As $d_0$ grows from $0.1$ km to $0.5$ km, the primary channel gain $g_{0,dB}$ between the PT and PR decreases. To guarantee the target SNR $\gamma_{T,dB}$ at the PR, the PT increases its transmit power. This enhances the average SNR of the primary signals at the CT as shown in Fig. \ref{SNR_d0}. Based on Fig. \ref{error_d0} and Fig. \ref{SNR_d0}, the average SNR of the primary signals at the CT increases from around $10$ dB to $35$ dB. From the analysis in previous figures, the estimation error caused by the measure error of each SNR with the ML estimator is negligible when the average SNR of the primary signals at the CT is no less than $5$ dB. Thus, $\varepsilon$ of the ML estimator almost remains constant in this figure. The trend of $\varepsilon$ with the MB estimator can be similarly explained. Again from the analysis in previous figures, the estimation error with the ML estimator is smaller than that with the MB estimator when the average SNR of the primary signal at the CT is larger than or equal to $4$ dB. Thus, the curve of $\varepsilon$ with the ML estimator is below that with the MB estimator in this figure.

\begin{figure}[t!]
            \centering
            \includegraphics[scale=0.5]{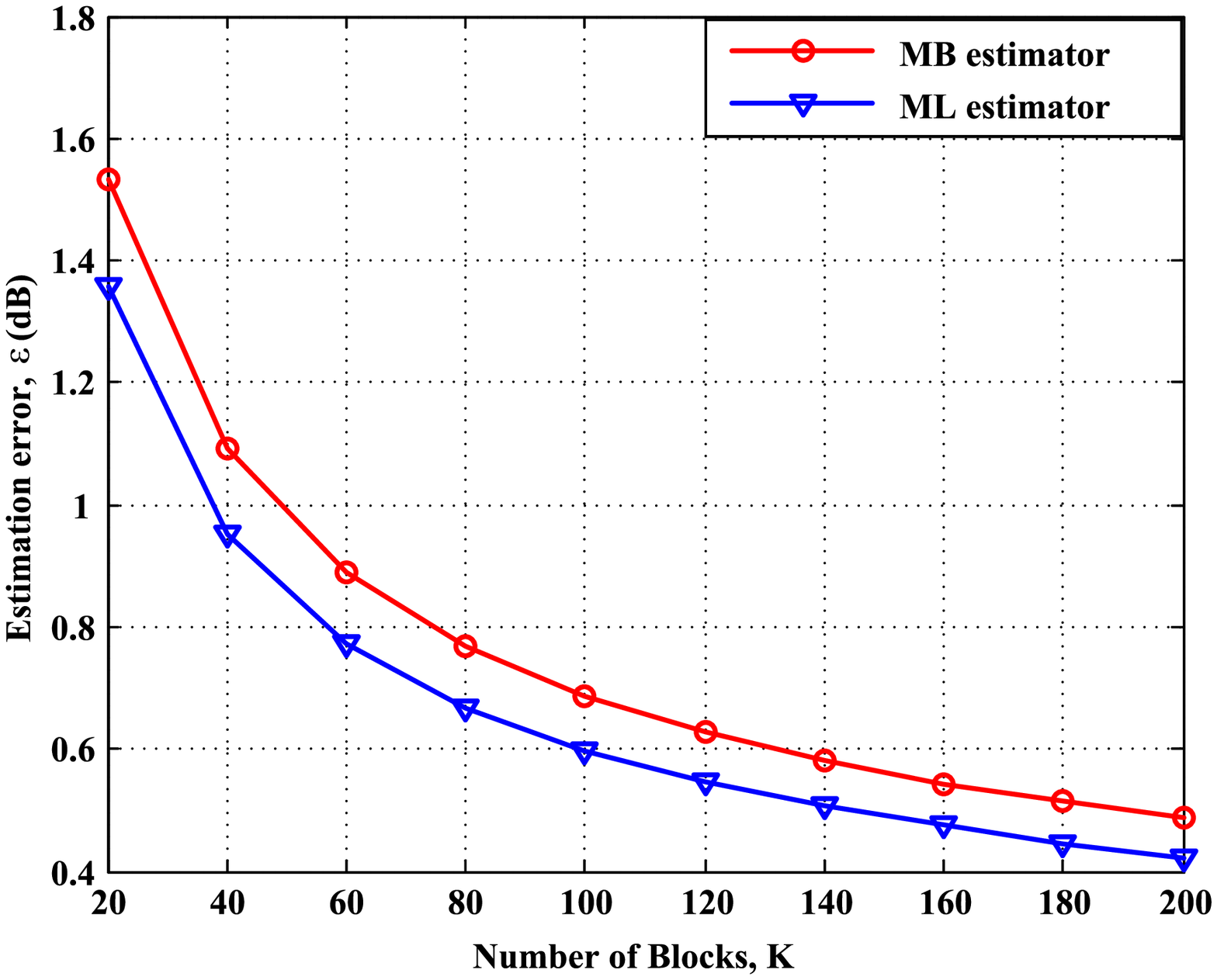}
            \caption{The estimation errors $\varepsilon$ of the ML estimator and the MB estimator, where ML estimator denotes the estimator of (\ref{ML_estimator}) and the MB estimator denotes the estimator in (\ref{MB_estimator}).}
            \label{error_K}
        \end{figure}

Fig. \ref{error_K} provides the estimation error $\varepsilon$ versus the number of the measured SNRs at the CT, i.e., $K$. In particular, the distance $d_0$ between the PT and PR is $0.25$ km and the distance $d_1$ between the PT and CT is $0.1$ km. From this figure, $\varepsilon$ of the ML estimator and the MB estimator monotonically decrease with the growth of $K$. This is because, a larger $K$ means that the ML estimator can utilize more measured SNRs at the CT, which provide more information of the primary channel gain $g_{0,dB}$. By adopting the ML criterion, the ML estimator is able to extract more information of $g_{0,dB}$ and outputs a more accurate estimation. This leads to a smaller estimation error with the ML estimator, and also verifies the results in Theorem 2. Besides, we observe that $\varepsilon$ of the ML estimator is smaller than that of the MB estimator. This is reasonable, since the average SNR of the measured primary signals is around $24$ dB from Fig. \ref{SNR_d1} when $d_0=0.25$ km and $d_1=0.1$ km. From the analysis of both Fig. \ref{error_d1} and Fig. \ref{SNR_d1}, $\varepsilon$ of the ML estimator is smaller than that of the MB estimator when the average SNR at the CT is no less than $4$ dB. Furthermore, we observe that, the gap of $\varepsilon$ with two different estimators is smaller than $0.1$ dB, which is negligible respect to the primary channel gain $g_{0,dB}$.

          \begin{figure}[t!]
            \centering
            \includegraphics[scale=0.5]{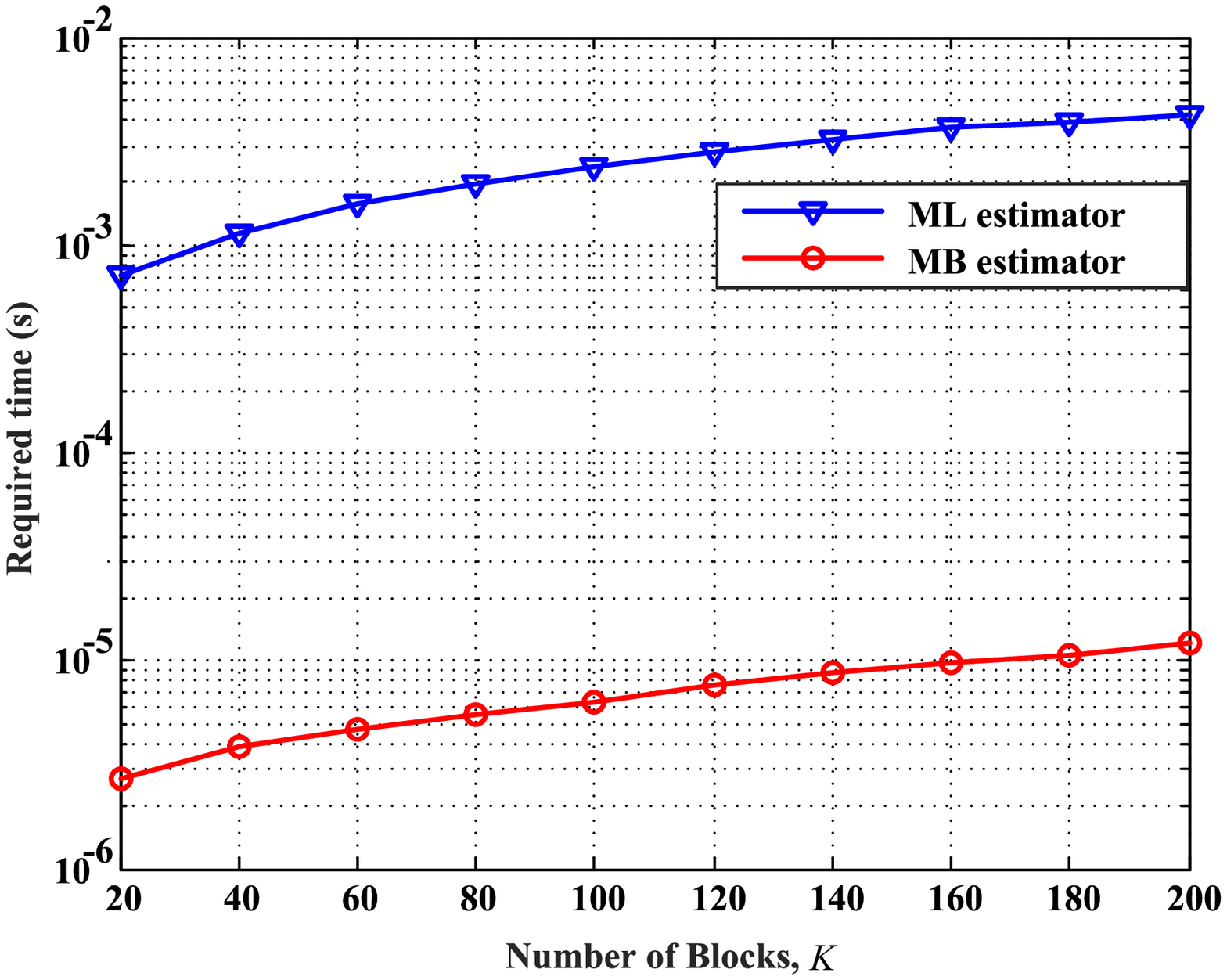}
            \caption{Comparison of the computation complexity with the ML estimator and the MB estimator. In particular, a smaller required times means a lower computation complexity.}
            \label{Time_K}
        \end{figure}

Accordingly, Fig. \ref{Time_K} compares the computation complexity with the ML estimator and the MB estimator. In particular, a smaller required time means a lower computation complexity. From this figure, the required time to obtain an estimation with the ML estimator is almost $100$ times of the required time with the MB estimator. This shows the advantages of the MB estimator over the ML estimator from the aspect of computational complexity and also verifies our analysis in the previous sections.

          \begin{figure}[t!]
            \centering
            \includegraphics[scale=0.5]{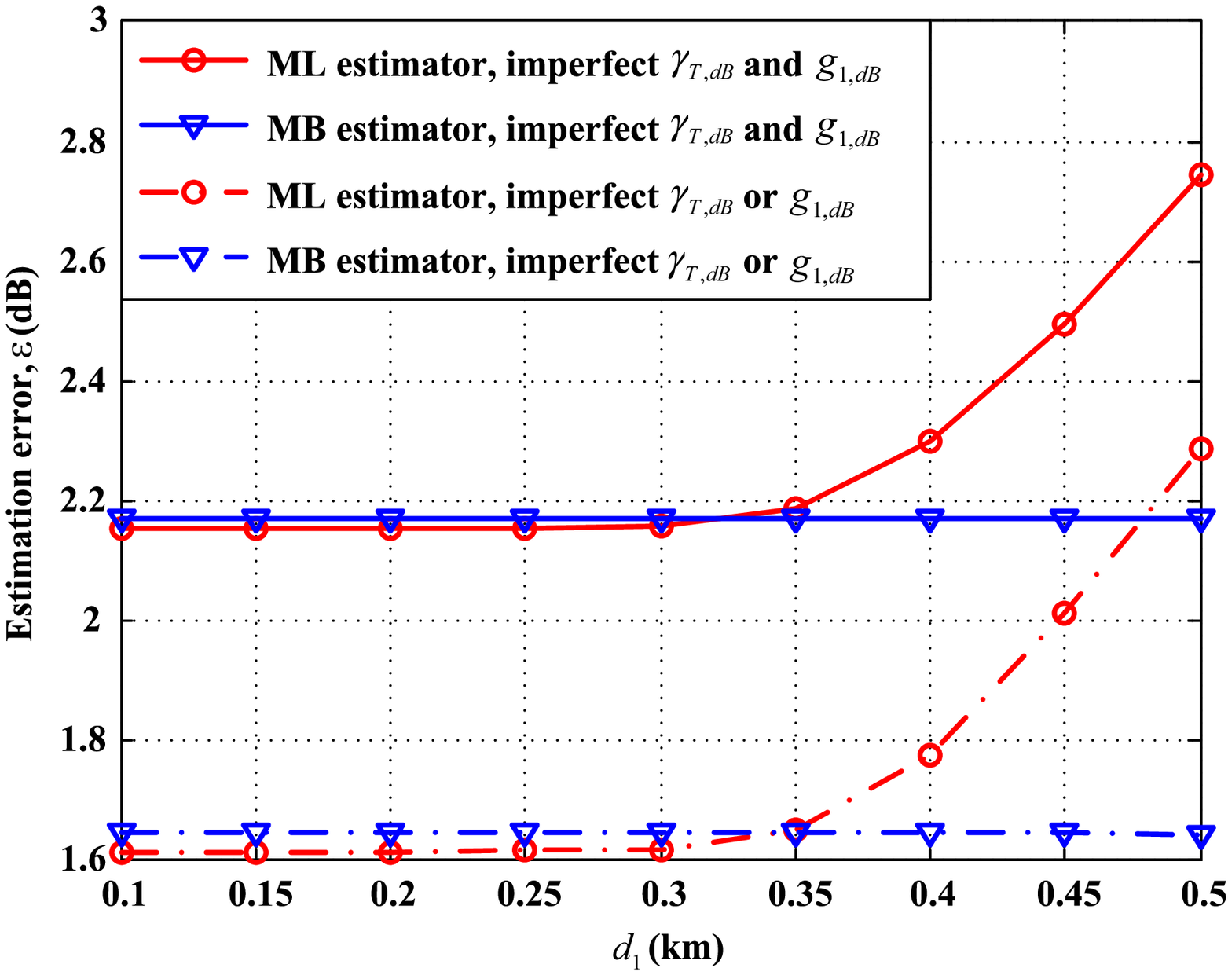}
            \caption{The estimation errors $\varepsilon$ of the ML estimator and the MB estimator with imperfect information of $\gamma_{T,dB}$ or/and $g_{1,dB}$, versus the distance $d_1$ between the PT and the CT. In particular, Case I denotes that the error of $\gamma_{T,dB}$ or $g_{1,dB}$ is uniformly distributed within $[-3, 3]$ dB. Case II denotes that the errors of both $\gamma_{T,dB}$ and $g_{1,dB}$ are uniformly distributed within $[-3, 3]$ dB.}
            \label{error_d1_gap}
        \end{figure}

Fig. 9 provides the estimation error $\varepsilon$ with imperfect information of $\gamma_{T,dB}$ or/and $g_{1,dB}$, versus the distance $d_1$ between the PT and the CT. In particular, Case I means that the error of $\gamma_{T,dB}$ or $g_{1,dB}$ is uniformly distributed within $[-3, 3]$ dB. Since the impacts of imperfect $\gamma_{T,dB}$ and imperfect $g_{1,dB}$ on the estimation of $g_{0,dB}$ are symmetrical from the relation among $g_{0,dB}$, $\gamma_{T,dB}$, and $g_{1,dB}$ in (\ref{gamma_c_dB}), $\varepsilon$ of both estimators are the same when the error of $\gamma_{T,dB}$ or $g_{1,dB}$ is uniformly distributed within $[-3, 3]$ dB. Case II means that the errors of both $\gamma_{T,dB}$ and $g_{1,dB}$ are uniformly distributed within $[-3, 3]$ dB. From this figure, the estimation errors of both estimators in Case I and Case II are increased by around $1$ dB and $1.5$ dB respectively, compared with the case where both $\gamma_{T,dB}$ and $g_{1,dB}$ are perfect. This shows that the estimation errors of both estimators can be no larger than around $2.2$ dB, even when both $\gamma_{T,dB}$ and $g_{1,dB}$ are imperfect, demonstrating the robustness of the proposed two estimators. Furthermore, we observe that the gap of the estimation errors between both estimators is dramatically reduced when considering imperfect $\gamma_{T,dB}$ or/and $g_{1,dB}$. This indicates that the MB estimator is more robust than the ML estimator in terms of imperfect $\gamma_{T,dB}$ or/and $g_{1,dB}$.

          \begin{figure}[t!]
            \centering
            \includegraphics[scale=0.5]{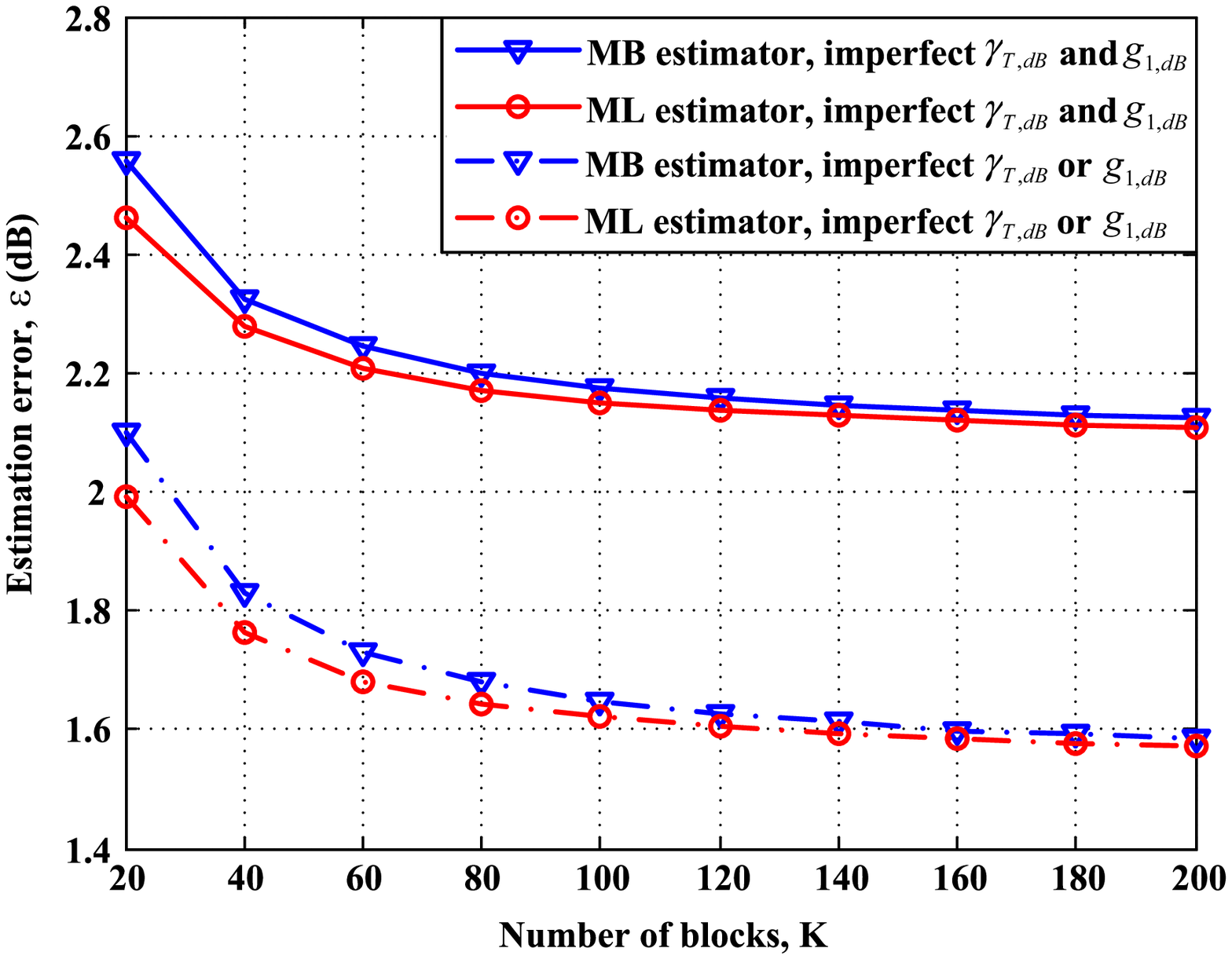}
            \caption{The estimation error $\varepsilon$ with imperfect information of $\gamma_{T,dB}$ or/and $g_{1,dB}$, versus the number $K$ of the measured SNRs at the CT. In particular, Case I denotes that the error of $\gamma_{T,dB}$ or $g_{1,dB}$ is uniformly distributed within $[-3, 3]$ dB. Case II denotes that the errors of both $\gamma_{T,dB}$ and $g_{1,dB}$ are uniformly distributed within $[-3, 3]$ dB.}
            \label{error_d1_gap}
        \end{figure}

Fig. 10 illustrates the estimation error $\varepsilon$ of $g_{0,dB}$ with imperfect information of $\gamma_{T,dB}$ or/and $g_{1,dB}$, versus the number $K$ of the measured SNRs at the CT. From this figure, $\varepsilon$ of both estimators in Case I and Case II is increased by around $0.8$ dB and $1.3$ dB respectively, compared with the case that both $\gamma_{T,dB}$ and $g_{1,dB}$ are perfect. Besides, as $K$ increases, $\varepsilon$ in both cases decreases. This indicates that, by increasing the number of the measured SNRs at the CT, the impacts of imperfect $\gamma_{T,dB}$ or/and $g_{1,dB}$ on the estimation of $g_{0,dB}$ can be reduced. In this way, we demonstrates the flexibility of the proposed two estimators.

\section{Conclusions}
In this paper, we proposed two estimators for the CT to estimate the primary channel gain, such that the CT is able to calculate the interference temperature of the primary system and achieve SS. In particular, we enabled the CT to sense primary signals and developed two estimators to obtain the primary channel gain. Numerical results show that the estimation errors of the ML estimator and the MB estimator can be as small as $0.6$ dB and $0.7$ dB, respectively. Besides, the ML estimator outperforms the MB estimator in terms of the estimation error if the SNR of the sensed primary signal at the CT is no smaller than $4$ dB. Otherwise, the MB estimator is superior to the ML estimator from the aspect of both computational complexity and estimation accuracy.


\section{Appendix}

\subsection{Derivation of the interference temperature with $g_0$}

Consider the scenario that the PT is transmitting data to the PR with a target SNR $\gamma_T$. If the PR is interfered by cognitive signals and the target SNR $\gamma_T$ at the PR cannot be satisfied, even when the PT works with the maximum transmit power $p_{\max}$, an outage of the primary transmission is claimed. In general, a specific $\gamma_T$ corresponds to a certain wireless service in the primary system and requires a preset maximum outage probability $\Theta$. Thus, an interference temperature $p_I$ is imposed on the transmit power of the CT to protect the primary transmission. Mathematically, we have \cite{Interference_T_C}
\begin{equation}
\Pr\left\{\frac{p_{\max}g_0|h_0|^2}{\sigma^2+p_I}<\gamma_T\right\} = \Theta,
\label{P_I_0}
\end{equation}
where $h_0$ denotes the small-scale block fading coefficient and $|h_0|$ follows a Rayleigh distribution with unit mean, and $g_0$ is the large-scale channel gain between the primary transceivers, and $\sigma^2$ represents the power of the AWGN. From (\ref{P_I_0}), the interference temperature $p_I$ can be calculated as
\begin{equation}
p_I=\frac{-p_{\max}g_0\ln(1-\Theta)}{\gamma_T}-\sigma^2.
\label{P_I}
\end{equation}

From (\ref{P_I}), the interference temperature $p_I$ is related to $p_{\max}$, $\gamma_T$, $\Theta$. In particular, $p_{\max}$ is a typical value of a PT and can be known as a prior knowledge at the CT, and $\gamma_T$ can be known at the CT by observing the MCS of the primary signal \cite{D_Tse}, and $\Theta$ corresponds to a specific $\gamma_T$ and can be known by the CT once $\gamma_T$ is obtained, and $\sigma^2$ is the power of the AWGN and is also available at the CT. Thus, the CT is able to calculate $p_I$ with $g_0$, $p_{\max}$, $\gamma_T$, $\Theta$.

\subsection{Proof of Theorem 1}
When $K$ is odd, the sample median is $\gamma^s_{c, dB, \frac{1}{2}}=\bar \gamma _{c, dB}\left(\frac{K+1}{2}\right)$. Note that, for these $K$ measured SNRs $\bar \gamma _{c, dB}(k)$ ($1\leq k\leq K$), we have $\bar \gamma _{c, dB}(i) \leq \bar \gamma _{c, dB}(j)$ for $1\leq i \leq j \leq K$. Then, the lower bound and the upper bound of $\gamma^s_{c, dB, \frac{1}{2}}$ can be denoted as
\begin{equation}
\grave{\gamma}_{c, dB}= \bar \gamma _c\left(\frac{K-1}{2}\right)
\label{Lower_bound_odd}
\end{equation}
and
\begin{equation}
\acute{\gamma}_{c, dB}= \bar \gamma _c\left(\frac{K+1}{2}+ 1\right),
\label{Upper_bound_odd}
\end{equation}
respectively.

Among the $K$ measured SNRs $\bar \gamma _{c, dB}(k)$ ($1\leq k\leq K$), the probabilities that $\bar \gamma _{c, dB}(k)$ is smaller than or equal to $\grave{\gamma}_{c, dB}$ and $\acute{\gamma}_{c, dB}$ are
\begin{equation}
\Pr\left\{\bar \gamma _{c, dB}(k)\le \grave{\gamma}_{c, dB} \right\} = \frac{{\frac{{K - 1}}{2}}}{K} = \frac{1}{2} - \frac{1}{K}
\label{Lower_bound_probability_1}
\end{equation}
and
\begin{equation}
\Pr\left\{\bar \gamma _{c, dB}(k)\le \acute{\gamma}_{c, dB} \right\} = \frac{{\frac{{K + 1}}{2} + 1}}{K} = \frac{1}{2} + \frac{3}{{2K}},
\label{Upper_bound_probability_1}
\end{equation}
respectively.

When $K$ is even, the sample median is $\gamma^s_{c, dB, \frac{1}{2}}=\frac{\bar\gamma _{c, dB}\left(\frac{K}{2}\right)+\bar\gamma _{c, dB}\left(\frac{K}{2}+1\right)}{2}$. Then, the lower bound and the upper bound of $\gamma^s_{c, dB, \frac{1}{2}}$ can be denoted as
\begin{equation}
\grave{\gamma}_{c, dB}= \bar\gamma_c\left(\frac{K}{2}\right),
\label{Lower_bound_odd}
\end{equation}
and
\begin{equation}
\acute{\gamma}_{c, dB}= \bar\gamma _c\left(\frac{K}{2}+1\right),
\label{Upper_bound_odd}
\end{equation}
respectively.

Among the $K$ measured SNRs $\bar\gamma _{c, dB}(k)$ ($1\leq k\leq K$), the probabilities that $\bar \gamma _{c, dB}(k)$ is smaller or equal to $\grave{\gamma}_{c, dB}$ and $\acute{\gamma}_{c, dB}$ are
\begin{equation}
\Pr\left\{\bar\gamma _{c, dB}(k)\le \grave{\gamma}_{c, dB} \right\} = \frac{{\frac{{K}}{2}}}{K} = \frac{1}{2}
\label{Lower_bound_probability_2}
\end{equation}
and
\begin{equation}
\Pr\left\{\bar\gamma _{c, dB}(k)\le \acute{\gamma}_{c, dB} \right\} = \frac{{\frac{{K}}{2}}+1}{K} = \frac{1}{2}+\frac{1}{K},
\label{Upper_bound_probability_2}
\end{equation}
respectively.

Based on (\ref{Lower_bound_probability_1}), (\ref{Upper_bound_probability_1}), (\ref{Lower_bound_probability_2}), and (\ref{Upper_bound_probability_2}), for any $K$, we have
\begin{equation}
  \frac{1}{2} - \frac{1}{K} \le \Pr\left\{\bar\gamma _{c, dB}(k)\le \grave{\gamma}_{c, dB}\right\} \le \frac{1}{2}
  \label{Lower_bound_probability_3}
\end{equation}
and
\begin{equation}
\frac{1}{2} + \frac{1}{K} \le \Pr\left\{ \bar\gamma _{c, dB}(k)\le \acute{\gamma}_{c, dB} \right\} \le \frac{1}{2} + \frac{3}{{2K}},
\label{Upper_bound_probability_3}
\end{equation}
respectively.

When $K$ goes to the infinity, i.e., $K \to \infty$, (\ref{Lower_bound_probability_3}) and (\ref{Upper_bound_probability_3}) become
\begin{equation}
  \frac{1}{2} - {\mu _0} \le F_{\Gamma_c}\left( \grave{\gamma}_{c, dB}\right) \le \frac{1}{2}
\end{equation}
and
\begin{equation}
\frac{1}{2} + {\mu _0} \le F_{\Gamma_c}\left( \acute{\gamma}_{c, dB} \right) \le \frac{1}{2} + \frac{3}{2}{\mu_0},
\end{equation}
where $\mu_0$ is defined as $\mu_0=\underset{K \to \infty}{\lim}\frac{1}{K}$ and is an arbitrarily small and positive value.

Since the median $\gamma_{c, dB, \frac{1}{2}}$ satisfies $F_{\Gamma_c}\left( \gamma_{c, dB, \frac{1}{2}}\right) = \frac{1}{2}$, we have
\begin{equation}
\Pr\left\{ {F_{\Gamma_c}\left( \gamma_{c, dB, \frac{1}{2}} \right) - F_{\Gamma_c}\left( \grave{\gamma}_{c, dB}\right) < {\mu_0}} \right\} = 1
\label{Combine_1}
\end{equation}
and
\begin{equation}
\Pr\left\{ {F_{\Gamma_c}\left( \gamma_{c, dB, \frac{1}{2}} \right) - F_{\Gamma_c}\left( \acute{\gamma}_{c, dB}\right) > \frac{3\mu_0}{2}} \right\} = 1,
\label{Combine_2}
\end{equation}
respectively.

Combining (\ref{Combine_1}) and (\ref{Combine_2}), we obtain
\begin{equation}
\begin{split}
&\Pr\left\{-\frac{3\mu_0}{2}< F_{\Gamma_c}\left( \gamma_{c, dB, \frac{1}{2}} \right)-F_{\Gamma_c}\left( \acute{\gamma}_{c, dB}\right)\right. \\
&\quad \quad <\left. F_{\Gamma_c}\left(\gamma_{c, dB, \frac{1}{2}}\right)-F_{\Gamma_c}\left( \grave{\gamma}_{c,dB} \right)<\mu_0 \right\}=1
\end{split}
\label{three_SNR_comapre_1}
\end{equation}

Since $F_{\Gamma_c}(\gamma_{c, dB})$ is a continuous and monotonically increasing function of $\gamma_{c, dB}$, the following inequations hold,
\begin{align} \nonumber
& \quad F_{\Gamma_c}\left( \gamma_{c, dB, \frac{1}{2}} \right)-F_{\Gamma_c}\left( \acute{\gamma}_{c, dB}\right)\\ \nonumber
&\leq  F_{\Gamma_c}\left( \gamma_{c, dB, \frac{1}{2}} \right)-F_{\Gamma_c}\left( \gamma^s_{c, dB, \frac{1}{2}}\right) \\
&\leq F_{\Gamma_c}\left( \gamma_{c, dB, \frac{1}{2}} \right)-F_{\Gamma_c}\left( \grave{\gamma}_{c, dB}\right).
\label{three_SNR_comapre_2}
\end{align}

Based on (\ref{three_SNR_comapre_1}) and (\ref{three_SNR_comapre_2}), we obtain
\begin{equation}
\Pr\left\{ { - \frac{3}{2}{\mu_0} < F_{\Gamma_c}\left( \gamma_{c, dB, \frac{1}{2}} \right)-F_{\Gamma_c}\left( \gamma^s_{c, dB, \frac{1}{2}}\right) < {\mu_0}} \right\}=1.
\end{equation}

Then, we can always find an arbitrarily positive and small $\mu _1$ satisfying
\begin{equation}
\Pr\left\{ - \mu_1 <  \gamma_{c, dB, \frac{1}{2}}- \gamma^s_{c, dB, \frac{1}{2}} < {\mu_1} \right\} = 1.
\label{gamma_c_cB_approach}
\end{equation}

Based on (\ref{gamma_c_cB_approach}) and the relations between $g_{0,dB}$ and $\gamma_{c, dB}$ in (\ref{gamma_c_dB}), we can always find an arbitrarily positive and small $\mu$ satisfying
\begin{equation}
\Pr\left\{ - \mu < g_{0,dB} - \hat{g}_{0,dB} < \mu \right\} = 1,
\end{equation}
which can be rewritten as
\begin{equation}
\Pr\left\{ {\left| {g_{0,dB}} - \hat{g}_{0,dB} \right| < \mu} \right\} = 1.
\end{equation}

This completes the proof of Theorem 1.

\end{document}